\documentclass[11pt]{article}\usepackage{amsfonts,amsmath,latexsym,amsthm,amssymb,euscript,eufrak,graphicx,units,mathrsfs, color,setspace}

\usepackage{authblk}

\usepackage{amsmath}
\usepackage{amsfonts}
\usepackage{amssymb}

\usepackage{amsthm}
\usepackage{graphicx}
\usepackage[T1]{fontenc}     
\usepackage{graphicx}
\usepackage{epstopdf}
\usepackage{graphics}
\usepackage{subfigure}
\usepackage{psfrag}
\graphicspath{{./Figures/}}
\usepackage{url}
\usepackage{hyperref}
\usepackage{textcomp}

\newtheorem{prop}{Proposition}[section]
\newtheorem{theorem}{Theorem}[section]
\newtheorem{cor}{Corollary}[section]

\newtheorem{lemma}{Lemma}[section]
\newtheorem{remark}{Remark}[section]
\newtheorem{example}{Example}[section]

\numberwithin{equation}{section}

\makeatletter \@addtoreset{equation}{section}

\AtBeginDocument{

}

\def\R{\mathbb{R}}

\def\E{\mathbb{E}}
\def\P{\mathbb{P}}
\def\real{\mathbb{R}}

\usepackage{geometry}

\usepackage{titlesec}

\titleformat{\section}
    {\normalfont\Large\scshape}{\thesection }{.5em}{\vspace{.5ex}}[\titlerule]

\allowdisplaybreaks

\title{Analysis of the Risk-Sharing Principal-Agent problem through the Reverse-H\"older inequality}

\author[1]{Jessica Martin\footnote{jessica.martin@insa-toulouse.fr}}

\author[1]{Anthony R\'eveillac\footnote{anthony.reveillac@insa-toulouse.fr}}

\affil[1]{INSA de Toulouse\\ IMT UMR CNRS 5219\\ Universit\'e de Toulouse\\ 135 Avenue de Rangueil 31077 Toulouse Cedex 4 France}

\begin{document}
\allowdisplaybreaks
\maketitle

\begin{abstract}
In this paper we provide an alternative framework to tackle the first-best Principal-Agent problem under CARA utilities. This framework leads to both a proof of existence and uniqueness of the solution to the Risk-Sharing problem under very general assumptions on the underlying contract space. Our analysis relies on an optimal decomposition of the expected utility of the Principal in terms of the reservation utility of the Agent and works both in a discrete time and continuous time setting. As a by-product this approach provides a novel way of characterizing the optimal contract in the CARA setting, which is as an alternative to the widely used Lagrangian method, and some analysis of the optimum.
\end{abstract}

\let\thefootnote\relax


\section{Introduction}
\label{section:introduction}

Many economic situations in the areas of optimal contracting or incentive policy design can be gathered under the so-called Principal-Agent formulation where an Agent is asked to perform an action on behalf of a Principal in exchange for a wage. This formulation has known a renewed interest since the seminal paper \cite{HM87} by Holmstr\"om and Milgrom in 1987 which brought to light a new method to solve a particular type of Principal-Agent problem in continuous time. Since then, this approach has been expanded.  For instance, some notable contributions include the works by Schattler and Sung in \cite{SSung93} and \cite{SSung97}, the works of Sung in \cite{Sung95}, and more recently the contribution \cite{Sannikov08} of  Sannikov. As we just mentioned the focal point of these works is one particular type of Principal-Agent problem, the Moral Hazard case, and the amount of interest for this problem is on par with its realism. Indeed,  it involves optimal contracting in a situation where the Principal cannot control and monitor the Agent's actions. This matches reasonably with  real-world applications. An investor who delegates his portfolio management to a banker cannot expect to control the banker's actions. A group of shareholders is not able to monitor every doing of the CEO of its company. It is thus said that the Moral Hazard problem models optimal contracting under partial information. \\

To measure the impact of this asymmetry of information, it is enlightening to compare the optimal contract with the optimal Risk-Sharing rule. This rule comes about in a full information situation, when the Principal dictates the Agent's actions and guarantees a reservation utility to him. Specific literature on this problem, which is called the first-best Principal-Agent problem, includes important works such as  \cite{CZ12}, \cite{CWZ05} and \cite{Muller98}. However, it is overall more scarce than the literature on the Moral Hazard problem. There are at least two reasons behind this. The first important reason is that the first-best problem is also a so-called Risk-Sharing problem. Risk-Sharing problems are documented in several fields of economics and insurance (but in another context than the one Principal and one Agent problem) and more general literature exists using for example risk-measure approaches.  We refer the reader to Section 1 of \cite{Embrechts17} for a review of the existing literature on the Risk-Sharing question. 
 Another reason for a scarce literature is that the first-best Principal-Agent problem is less realistic than other variations on the Principal-Agent problem such as Moral Hazard, since it does not fit strictly speaking in an incentive policy design framework. This lack of realism means that the problem is less subject to extensions than the Moral Hazard problem or other variations on the Principal-Agent problem. \\

When analyzing Principal-Agent problems of any form two main questions arise. Does a solution to the set problem actually exist? And if so, can we obtain some form of characterization for it? The existence question is fundamental as without any such guarantees trying to characterize a solution is pointless. For this reason, obtaining existence results has aroused a considerable amount of interest across the past few decades and for different types of Principal-Agent problems. In 1987, Page considered the Moral Hazard problem  in \cite{Page87} and was able to prove existence. To do so, he used a set of assumptions including a compact set for the agent's actions, bounded wages, and a compact set for monetary outcomes. The same year and in their seminal paper \cite{HM87}, Hölmstrom and Milgrom proved existence in a single period Moral Hazard problem with a discrete (and thus bounded) wealth process, a compact set of feasible actions, and CARA utilities.  They then build a continuous time model that approaches the discrete one. Shortly after, Schättler and Sung extended the continuous time model of \cite{HM87} and established in \cite{SSung93} a set of necessary and sufficient conditions for existence of solutions in a continuous time setting. One such condition is that the agent's actions belong to a compact set but they also suppose that the wage be bounded below before proving that their hypothesis can be slightly lifted in order to include a few more general wages such as linear ones. A more recent work is that of Jewitt, Kadan and Swinkels in \cite{JKS} where one key assumption is that the wealth process is supposed bounded. 
Many works thus consider quite restrictive hypothesis to obtain existence. However, a few papers aim to provide more general existence conditions. These include the work of Kadan, Reny and Swinkels in \cite{KRS} for the Moral Hazard and Adverse Selection setting using only a bounded wage assumption. The Adverse Selection case in continuous-time is also tackled by Carlier in \cite{Carlier01} using calculus of variations and h-convexity. This work is further extended by Carlier and Zhang in \cite{Carlier19} using calculus of variations once again along with an important assumption to obtain the required compactness. As a consequence to this assumption, the authors consider mainly Lipschitz utilities in their examples. \\

This literature analysis leads us to several important conclusions. A first conclusion is that the first-best problem is a benchmark problem for Principal-Agent analysis yet specific literature on it is scarce. Furthering the literature on this problem and on its possible extensions will help gain better understanding of Principal-Agent problems as a whole. Another conclusion is that a key setting for Principal-Agent problems involves CARA utility and a discrete or gaussian wealth process. In the latter case, the model becomes the so-called LEN model. This model is for example discussed in \cite{Muller98} by Muller who characterizes the solutions to the problem in a continuous time first-best setting. Finally, establishing solution existence is both key and non-trivial for Principal-Agent problems before going ahead with any form of characterization. Indeed due to the stochastic nature of the problems, the underlying contract spaces are in most cases non-compact. Additional assumptions on models are then needed to obtain compactness, particularly in discrete time settings where results from the field of optimal stochastic control cannot be used, to be able to apply standard existence theorems from the field of optimization. The literature therefore includes few existence results for Principal-Agent problems in general settings. \\

In this paper we aim to bring to light a result that is at the crossroads of these three conclusions. Indeed, we consider a Principal who owns a firm (or a portfolio) whose wealth is subject to uncertainty and an Agent to whom a wage is offered in exchange for a participation to the firm (and we also refer the reader to \cite{CZ12} and \cite{CWZ05} for more details). Their respective utilities $U_P$ and $U_A$ are CARA utilities defined as:
$$ U_P(x) = -\exp(-\gamma_P x) \quad \text{and} \quad U_A(x) = -\exp(-\gamma_Ax),$$
where $\gamma_P$ and $\gamma_A$ are two fixed risk aversion coefficients. In a single period setting, in order to reduce his exposure to uncertainty, the Principal hires the Agent at time $t=0$ in a take-it-or-leave-it contract in which (if accepted) the Agent is asked to produce an effort $a \geq 0$ at time $t=0$ in return for the payment of a wage $W$ at time $t=1$. The wealth of the Agent then becomes $X^a$ at time $t=1$ 
with
$$ X^a = x_0 + a + B,$$
where $x_0$ is the initial wealth at time $t=0$ and $B$ is a square-integrable random variable modeling the stochastic exposure of the Principal. Note that in this simple model, we assume that the Principal fully observes both the outcome $X^a$ (at time $t=1$) and the action $a$ of the Agent, meaning that the Principal actually dictates this action to the Agent. \footnote{A typical example of such a situation is when the Agent is the Principal himself, meaning that as a manager of the firm, the Principal decides his level of work and the salary he pays himself for it. Note that as the Principal decides of the action of the agent, this action must be a positive effort in the firm as $a<0$ would mean that the Agent would sabotage the firm, which of course does not make any sense. Another example is a monopoly type situation. } To model the cost of effort for the Agent, we introduce a function $\kappa$ defined on $\R_+$ and chosen to be strictly convex, continuous and non-decreasing. A simple example of such a function is the quadratic cost function which, for a fixed constant $K>0$, is defined for any $x$ in $\R_+$ as $\kappa(x) = K\frac{|x|^2}{2}.$ 
\\\\\
\noindent
An important remark is that the contract (which from now on will be modeled by the pair (wage, action)$=(W,a)$) is a take-it-or-leave-it contract that will be accepted by the Agent if a Participation Constraint (PC) condition (or reservation utility constraint) given below is satisfied : 
\begin{equation}
\label{eq:PCIntro}
\E\left[U_A\left(W-\kappa(a)\right)\right] \geq U_A(y_0),
\end{equation}
where $y_0\geq 0$ represents the level of requirement for the Agent to accept the contract.\\\\
\noindent
In this setting, we analyze the first-best problem which simply writes as~: 
\begin{equation}
\label{eq:ProbIntro}
\sup_{(W, a) \textrm{ subject to } (\ref{eq:PCIntro})} \E\left[U_P(X^a-W)\right], 
\end{equation}
and which can also be written in a continuous time setting. One classical way of proving existence of a solution to such a problem is to use a variational approach where we find a topology of the underlying contract space that ensures upper semi-continuity, concavity and coercivity of the Principal's expected utility, whilst rendering the contract space convex and closed.  The sticking point here is often coercivity and more particularly coercivity in $W$. This may be ensured in some cases when $B$ is bounded. However it is difficult to do so under more general  assumptions.   \\

In this paper we provide an alternative proof of existence and uniqueness of solutions to Problem (\ref{eq:ProbIntro})
 with next to no further assumptions. To do so, we exploit the properties of the exponential function in the utilities to first decompose the Principal's expected utility in terms of that of the Agent. Using this decomposition and the so-called Reverse Hölder inequality, we are able to upper bound the Principal's expected utility by some constant depending only on the model's parameters. We then derive an optimal decomposition, relying on the Reverse Hölder equality condition, which provides us with a unique admissible contract attaining the upper bound. As a by-product this allows us to shed a new light on the Borch rule for Risk-Sharing. We thus have an existence and uniqueness proof for the Risk-Sharing problem under CARA utilities along with a characterization of the optimum. We believe that the strength of this method lies in the generality of the settings that it allows us to consider. Indeed we are able to deal with both the discrete and continuous time case, actions that belong to $\mathbb{R}^+$ or even to any compact subset of $\mathbb{R}^+$, a general effort cost function $\kappa$, and more importantly any form on random variable / random process that has exponential moments (this last assumption is a purely technical as we are optimizing exponential expected utilities). Furthermore, the proof does not rely on any a priori intuition on the form of the optimal contract. We thus provide an existence and uniqueness result for the first-best Principal Agent problem under CARA utilities in more general settings that those generally considered in the literature. As a by-product to this approach we also obtain through our decomposition of the Principal's expected utility a revisit of the well-known Borch rule for Risk-Sharing. Finally, we provide some further analysis on the effect of different parameters on the Principal's take-home utility. Indeed, we analyze the effect of enforcing a sub-optimal action as well as study the question of a Principal who has to choose between two different Agents. \\

We proceed as follows. First in Section \ref{section:discrete} we tackle the single-period problem as described in the introduction. Then in Section \ref{section:an}, we proceed with some contract comparison in the single-period setting. Finally in Section \ref{section:exten} we discuss on some extensions, including the continuous time problem, before collecting some of the lengthier proofs of our results in Section \ref{section:proofs} and concluding in Section \ref{section:conc}. 

\section{The single-period Risk-Sharing problem}
\label{section:discrete}

In the following we discuss Problem (\ref{eq:ProbIntro}) as presented in the introduction. To do so we exploit a key result called the Reverse H\"older inequality. As its name suggests it is closely linked to the more well-known H\"older inequality. We state the result below. 
\begin{prop}[Reverse H\"older inequality]
 \label{prop:RH}
 Let $p \in (1 , + \infty]$. Let $F$ and $G$ be two random variables such that $G\neq 0$, $\P$-a.s.. Then : 
 \begin{itemize}
 \item[(i)] The Reverse H{\"o}lder inequality holds, that is,
 $$\E[|F \times G|] \geq \E\left[|F|^{\frac1p}\right]^p \times \E\left[|G|^{\frac{-1}{p-1}} \right]^{-p+1}.$$ 
 \item[(ii)] In addition, the inequality is an equality, that is,
 $$\E[|F \times G|] = \E\left[|F|^{\frac1p}\right]^p \times \E\left[|G|^{\frac{-1}{p-1}} \right]^{-p+1}$$
 if and only if there exists some constant (that is non-random) $\alpha \geq 0$ such that $|F| = \alpha |G|^{-\frac{p}{p-1}}$.\\
 \end{itemize}
 \end{prop}
 
In order to proceed with our analysis of Problem (\ref{eq:ProbIntro}), we define the set of admissible contracts for the Risk-Sharing problem $C^{adm}$ as follows : 
 $$C^{adm} = \left\{ (W,a) \in L^2(\Omega) \times \mathbb{R}_+ \quad \text{satisfying} \quad (\ref{eq:PCIntro}) \right\},$$
 where $L^2(\Omega)$ is the set of square-integrable random variables. Also, for a given $a$ in $\mathbb{R}_+$ we define the set of related admissible wages : 
 $$ \mathcal{W}(a) = \left\{ W \in L^2(\Omega), \quad \text{such that} \quad (W,a) \quad \text{in} \quad C^{adm} \right\}.$$
 With these notations in mind, we rewrite Problem (\ref{eq:ProbIntro}) as : 
 
 \begin{equation}
\label{eq:ProbRS}
\sup_{(W, a) \in C^{adm}} \E\left[U_P(X^a-W)\right], 
\end{equation}

or alternatively : 
 \begin{equation}
\label{eq:ProbRS2}
\sup_{W \in \mathcal{W}(a) }\sup_{a \in \mathbb{R}^+} \E\left[U_P(X^a-W)\right]. 
\end{equation}

Our aim is to prove existence of a solution to this problem. To do so, we use the inequality of Proposition \ref{prop:RH} and the multiplicative nature of the Principal's utility in order to obtain an attainable upper bound. As a first step, we apply the Reverse H\"older inequality and a decomposition of the Principal's expected utility. 
\begin{theorem}
\label{th:main1}
Let $a$ in $\R^+$. We have :
\begin{itemize}
\item[(i)] For any $W$ in $\mathcal{W}(a)$, 
\begin{align}
\label{eq:allW}
&\E\left[U_P(X^a-W)\right] \nonumber\\
&\leq \E\left[\left|U_P\left(X^a -  \kappa(a) \right)\right|^{\frac{\gamma_A}{\gamma_A + \gamma_P}}\right]^{\frac{\gamma_A+\gamma_P}{\gamma_A}} \times \E\left[ U_A\left(W- \kappa(a) \right)\right]^{-\frac{\gamma_P}{\gamma_A}}. 
\end{align}

\item[(ii)] For $W$ in $\mathcal{W}(a)$. The following conditions are equivalent :
\begin{itemize}
\item[(ii\textquotesingle
)] 
\begin{align}
\label{eq:WLin}
&\E\left[U_P(X^a-W)\right] \nonumber\\
&= \E\left[\left|U_P\left(X^a -  \kappa(a) \right)\right|^{\frac{\gamma_A}{\gamma_A + \gamma_P}}\right]^{\frac{\gamma_A+\gamma_P}{\gamma_A}} \times \E\left[ U_A\left(W- \kappa(a) \right)\right]^{-\frac{\gamma_P}{\gamma_A}}
\end{align}
\item[(ii\textquotesingle\textquotesingle
)] $(a,W)$ satisfies the Borch rule, that is, there exists $\alpha$ in $\real_+^*$ such that : 
\begin{equation}
\label{eq:BRNotre}
\frac{U_P'\left(X^a -W\right)}{U_A'\left(W-  \kappa(a) \right)}=\alpha, \quad \P-a.s..
\end{equation} 
\item[(ii\textquotesingle\textquotesingle\textquotesingle)] The wage $W$ is of the form : 
$$ W= \frac{\gamma_P}{\gamma_P+\gamma_A} X^a + \beta, \quad \text{with} \; \beta \in \R,$$ 
such that $(W,a) \in \mathcal{C}^{adm}.$
\end{itemize}
\end{itemize} 
\end{theorem}
\begin{proof}
See Section \ref{section:proofdecomp}.\\
\end{proof}

\begin{remark}
\begin{itemize}
\item[(i)]The Borch rule for Risk-Sharing, derived by Karl Borch in \cite{Borch62} and \cite{Borch74}, appears here in (\ref{eq:BRNotre}) as a condition for equality between the Principal's expected utility and its decomposition. This sheds a new light onto the rule in the CARA utility case. Indeed, it only allows for contracts that enable an isolation of the effect of the Agent's expected utillity in the Principal's expected utility. We will see further on that the Borch rule is a necessary optimality condition for the Risk-Sharing problem just like it is when using the Lagrangian method.
\item[(ii)] When analyzing the first-best problem it seems intuitive that the Principal and the Agent's utilities should be of opposite effect : the Principal should want to maximize his utility whilst minimizing that of the Agent. This is encompassed in (\ref{eq:allW}) : due to the negative power $-\frac{\gamma_P}{\gamma_A}$, the upper bound increases as the Agent's utility decreases. 
\end{itemize} 
\end{remark}

Now that we have a decomposition of the Principal's expected utility and a condition for it to hold exactly (Borch rule), we turn to further exploiting this decomposition and the appearance of the Agent's expected utility in order to obtain a bound that is free of $W$ and $a$. We do this in two stages in the following two Propositions. 

\begin{prop}\label{PCoptim}
Let $a$ in $\R_+$. For any $W$ in $\mathcal{W}(a)$, 
\begin{align*}
&\E\left[\left|U_P\left(X^a -  \kappa(a) \right)\right|^{\frac{\gamma_A}{\gamma_A + \gamma_P}}\right]^{\frac{\gamma_A+\gamma_P}{\gamma_A}} \times \E\left[ U_A\left(W-  \kappa(a) \right)\right]^{-\frac{\gamma_P}{\gamma_A}} \\
& \leq  U_P(-y_0) \times \inf_{\tilde a \in \mathbb{R}_+} \E\left[\left|U_P\left(X^{\tilde a} -  \kappa(\tilde a) \right)\right|^{\frac{\gamma_A}{\gamma_A + \gamma_P}}\right]^{\frac{\gamma_A+\gamma_P}{\gamma_A}}
\end{align*}
\end{prop}
\begin{proof}
Apply the participation constraint (\ref{eq:PCIntro}) to the right hand term and optimize in $a$ in the left hand term.
\end{proof}

It thus remains to perform the optimization in $a$. To do so, we introduce the  following two notations : 
$$\widetilde{\kappa}(p) := \sup_{x \in \R_+} \left( px-\kappa(x)    \right), \text{ for any } p \geq 0,$$
and
$$\kappa^*(p) := \text{argsup}_{x \in \R_+} \left( px-\kappa(x)    \right), \text{ for any } p \geq 0.$$ 
$\tilde{\kappa}$ is the Legendre transform of $\kappa$, and $\kappa^*$ is its related argument. These are well defined due to the convexity of $\kappa$. We use these two notations in the following Proposition to perform the optimization in $a$ and obtain our upper bound. 
\begin{prop}\label{prop:upperbound}
For any $a$ in $\R_+$ it holds that :
\begin{align*}
\E\left[\left|U_P\left(X^a - \kappa(a) \right)\right|^{\frac{\gamma_A}{\gamma_A + \gamma_P}}\right]^{\frac{\gamma_A+\gamma_P}{\gamma_A}} & \geq \E\left[\left|U_P\left(X^{a^*} -  \kappa(a^*) \right)\right|^{\frac{\gamma_A}{\gamma_A + \gamma_P}}\right]^{\frac{\gamma_A+\gamma_P}{\gamma_A}}\\
& = \exp\left(-\gamma_P (x_0 + \tilde{\kappa}(1))\right) \E\left[  \exp\left(\frac{-\gamma_P \gamma_A}{\gamma_P+\gamma_A} B \right)\right]^{\frac{\gamma_A + \gamma_P}{\gamma_A}},
\end{align*}
where $a^* := \kappa^*(1)$. 
We thus have for any $(W,a)$ in $C^{adm}$ : 
\begin{align*}
&\E\left[\left|U_P\left(X^a -   \kappa(a) \right)\right|^{\frac{\gamma_A}{\gamma_A + \gamma_P}}\right]^{\frac{\gamma_A+\gamma_P}{\gamma_A}} \times \E\left[ U_A\left(W- \kappa(a)  \right)\right]^{-\frac{\gamma_P}{\gamma_A}} \\
&\leq U_P(x_0-y_0 + \widetilde{\kappa}(1))  \E\left[  \exp\left(\frac{-\gamma_P \gamma_A}{\gamma_P+\gamma_A} B \right)\right]^{\frac{\gamma_A + \gamma_P}{\gamma_A}}.
\end{align*}
\end{prop}
\begin{proof}
See Section \ref{section:proofPropupperBound}.
\end{proof}

The combination of Theorem \ref{th:main1}, Proposition \ref{PCoptim} and Proposition \ref{prop:upperbound} allows us to exploit a decomposition of the Principal's expected utility, the Participation Constraint and an optimization in $a$ in order to upper bound the Principal's value function. The upper bound that we obtain is free of $W$ and $a$ and is key for our existence proof. Indeed, we are now able to show that this upper bound is attained for an admissible contract, which is in fact unique. This gives us our main result, which is simultaneously an existence, uniqueness, and characterization result for solutions to the first-best problem and is the object of the following key Theorem. 
\begin{theorem}[Existence, Uniqueness and Characterization]
\label{theo:main2}
\begin{itemize}
\item[(i)] Consider $(W,a)$ in any contract in $C^{adm}$, then it holds that : 
$$\E\left[U_P(X^a-W)\right] \leq U_P(x_0-y_0 + \widetilde{\kappa}(1))  \E\left[  \exp\left(\frac{-\gamma_P \gamma_A}{\gamma_P+\gamma_A} B \right)\right]^{\frac{\gamma_A + \gamma_P}{\gamma_A}}.$$

\item[(ii)]Now let $(W^*, a^*)$ be such that :  $$a^* = \kappa^*(1)$$ and 
 $$W^* = \frac{\gamma_P}{\gamma_P + \gamma_A} X^{a^*} + \beta^*,$$ 
$$ \beta^*:=  y_0+ \kappa(\kappa^*(1)) - \frac{\gamma_P}{\gamma_A+\gamma_P} \left(x_0+  \kappa^*(1) \right) - \frac{1}{\gamma_A}\ln\left( \mathbb{E}\left[\exp\left(\frac{-\gamma_P \gamma_A}{\gamma_P+\gamma_A} B\right)\right] \right).$$
Then $(W^*, a^*)$ is the only contact verifying (\ref{eq:BRNotre}) and saturating the participation constraint. 
\item[(iii)]Furthermore : 
 \begin{align*}
& \E\left[U_P\left(X^{a^*} - W^*\right)\right] =  U_P(x_0-y_0 + \widetilde{\kappa}(1))  \E\left[  \exp\left(\frac{-\gamma_P \gamma_A}{\gamma_P+\gamma_A} B \right)\right]^{\frac{\gamma_A + \gamma_P}{\gamma_A}}.
\end{align*}
Thus $(W^*, a^*)$ is the unique contract attaining the upper bound. It follows that for any $(W,a)$ in $C^{adm} :$
\begin{align*}
&\E\left[U_P\left(X^{a} - W\right)\right] \leq \E\left[U_P\left(X^{a^*} - W^*\right)\right], 
\end{align*}
and $(W^*, a^*)$ is the optimal contract for the first-best Principal-Agent problem. 
\end{itemize}
\end{theorem}
\begin{proof}
See Section \ref{section:ProofOpt}.
\end{proof}

\begin{remark}
Throughout this Section, we suppose that we wish to find an optimal action $a$ amongst the whole of $\R_+$. We may in fact also consider optimizations on compact subsets of $\R_+$. For $S$ such a subset, the optimal action is then :  $$ a^* = \text{argsup}_{a \in S} \quad a - \kappa(a) $$
\end{remark}

We may apply our key Theorem to a more specific Gaussian setting as follows.
\begin{example}
Consider a Gaussian setting where $B \sim \mathcal{N}(0,1),$ and $\kappa(x)~:=~K \frac{x^2}{2}$ for some $K$ in $\R_+$.  Then there exists a unique optimal contact for the Risk-Sharing problem. We set : 
$$ a^* := \kappa^*(1) = \frac{1}{K},$$
$$ W^* = \frac{\gamma_P}{\gamma_A + \gamma_P} X^{a^*} + \beta^*,$$
where the optimal $\beta^*$ has the following expression :  
$$ \beta^*:= \frac{ \gamma_A |\gamma_P|^2}{2|\gamma_A+\gamma_P|^2} + y_0+ \kappa(\kappa^*(1)) - \frac{\gamma_P}{\gamma_A+\gamma_P} \left(x_0+  \kappa^*(1) \right).$$
The contract $(W^*,a^*)$ is the optimal Risk-Sharing contract in this setting. The Principal's optimal expected utility is worth : 
$$\E\left[ U_P(X^{a^*} - W^*) \right] = U_P(x_0-y_0) \exp\left(\gamma_P \left(-\widetilde{\kappa}(1) + \frac{\gamma_P \gamma_A}{2(\gamma_A+\gamma_P)} \right) \right).$$
 Our main Theorem therefore includes the well-known results for the single period Gaussian setting discussed for example in \cite{CZ12}. We note that it even allows us to go further than \cite{CZ12} and fully specify the intercept $\beta^*$ in the wage rather than leaving it dependent on the Lagrange multiplier.
\end{example}

\begin{remark}
 Theorem \ref{theo:main2} extends existence of solutions beyond the bounded wealth process setting. It does so without any assumption on the form of the optimal contract and using an analytic inequality rather than calculus of variations for which we lack coercivity. It also provides an important proof of uniqueness. The result therefore completes pre-standing results on existence and uniqueness by allowing for a general setting (general wealth process, general cost of effort function etc.) in the CARA utility case. 
\end{remark}

This reasoning discussed in this Section generalizes to continuous time settings and this will be a focus of Section \ref{section:exten}. In the meantime we go over some further analysis of the Risk-Sharing problem that can be gleaned using the Reverse-Holder inequality.

\section{A comparison of contracts through the Reverse H\"older inequality}
\label{section:an}

In this Section we use the Reverse H\"older inequality given in Proposition \ref{prop:RH} to compare different contracts in the Risk-Sharing setting. Indeed, the multiplicative nature of the decomposition of the Principal's expected utility prompts us to compare ratios of utilities under different conditions. This analysis brings to light effects of some choices the Principal may wish to make. \\

We first study the effect of enforcing another action than the optimal one. The Principal may indeed wish to over or under work the Agent in some conditions and the following Proposition quantifies the effect of such a choice on the Principal's own utility. Note that the direction of inequality (\ref{eq:ineqWa}) may seem counter-intuitive at first but this is due to the negative sign of the expected utilities. 
\begin{prop}[Effect of enforcing a sub-optimal $a$]
\label{prop:effa}
Let $a$ be any positive action and $W$ be any contract in $\mathcal{W}(a)$. Let  $(W^*, a^*)$ be the optimal contract described in Theorem \ref{theo:main2}. We define the Agent's action ratio $R(a)$ as : 
$$ R(a) := \frac{\E\left[\left|U_P\left(X^a -  \kappa(a) \right)\right|^{\frac{\gamma_A}{\gamma_A + \gamma_P}}\right]}{\E\left[\left|U_P\left(X^{a^*} -  \kappa(a^*) \right)\right|^{\frac{\gamma_A}{\gamma_A + \gamma_P}}\right]} = \left( \frac{\exp(-\gamma_P( a - \kappa(a)))}{\exp(-\gamma_P( a^* - \kappa(a^*))}\right)^{\frac{\gamma_A}{\gamma_A + \gamma_P}}, $$
and the Participation Constraint ratio $\text{C}(W,a)$ as :
$$ \text{C}(W,a) :=  \frac{\E\left[ U_A\left(W-  \kappa(a) \right)\right]}{ \E\left[ U_A\left(W^*-  \kappa(a^*) \right)\right]} =  \frac{\E\left[ U_A\left(W-  \kappa(a) \right)\right]}{ U_A(y_0)}.$$
Then it holds that : 
\begin{align}
\label{eq:ineqWa}
\frac{\E\left[U_P\left(X^{a} - W\right)\right] }{ \E\left[U_P\left(X^{a^*} - W^*\right)\right]} &\geq R(a)^{\frac{\gamma_A+\gamma_P}{\gamma_A}} \times C(W,a)^{-\frac{\gamma_P}{\gamma_A}}, 
 \end{align}
 and 
\begin{align}
\label{eq:ineqa}
\E\left[U_P\left(X^{a} - W\right)\right] & \leq R(a)^{\frac{\gamma_A + \gamma_P}{\gamma_A}}  \times \E\left[U_P\left(X^{a^*} - W^*\right)\right].
 \end{align}
 These inequalities are strict as soon as $(W,a) \neq (W^*, a^*)$. Furthermore, when $(W,a)$ bounds the Participation Constraint, $C(W,a) = 1$ and we directly obtain \ref{eq:ineqa}. Finally $R(a) \geq 1.$
 \end{prop}
 \begin{proof}
The proof of (\ref{eq:ineqWa}) is a direct consequence to applying Reverse Holder to $\E\left[U_P\left(X^{a} - W\right)\right]$ and $\E\left[U_P\left(X^{a^*} - W^*\right)\right].$ Indeed, we obtain : 
$$\E\left[U_P\left(X^{a} - W\right)\right] \leq \E\left[\left|U_P\left(X^a -  \kappa(a) \right)\right|^{\frac{\gamma_A}{\gamma_A + \gamma_P}}\right]^{\frac{\gamma_A+\gamma_P}{\gamma_A}} \times \E\left[ U_A\left(W-  \kappa(a) \right)\right]^{-\frac{\gamma_P}{\gamma_A}},$$
and through optimality for $(W^*, a^*)$ : 
$$\E\left[U_P\left(X^{a^*} - W^*\right)\right] = \E\left[\left|U_P\left(X^{a^*} -  \kappa(a^*) \right)\right|^{\frac{\gamma_A}{\gamma_A + \gamma_P}}\right]^{\frac{\gamma_A+\gamma_P}{\gamma_A}}  \times \E\left[ U_A\left(W^*-  \kappa(a^*) \right)\right]^{-\frac{\gamma_P}{\gamma_A}}.$$
Taking the ratio of the two, noting that $\E[U_P(X^{a^*}-W^*)] \leq 0$ which changes the sign of the inequality, we obtain (\ref{eq:ineqWa}). To obtain (\ref{eq:ineqa}), note that through the Participation Constraint : 
$$\E\left[ U_A\left(W-  \kappa(a) \right)\right] \geq \E\left[ U_A\left(W^*-  \kappa(a^*) \right)\right],$$
and thus (as $U_A$ is negative)
$C(W,a) \leq 1$, implying that $C(W,a)^{-\frac{\gamma_P}{\gamma_A}} \geq 1.$ Finally, we obtain $R(a) \geq 1$ directly through Proposition \ref{prop:upperbound}.
\end{proof}
 
It is apparent through (\ref{eq:ineqWa}) that the ratio of the Principal's utilities splits into a product of two ratios. One of them ($C(W,a)$) transfers the effect of the Agent's participation constraint onto the utility of the Principal : indeed, $C(W,a)$ is maximized as soon as $(W,a)$ binds (and thus minimizes) the Agent's utility. The second ratio $R(a)$ further restricts the possible values for the Agent's value function and this no matter the chosen wage. Indeed, we see in (\ref{eq:ineqa}) that a choice of action $a$ may restrict quite consequently the possible attainable utilities for the Principal, no matter whether the associated wage binds the Participation Constraint or not. We illustrate this result in the following example. 

\begin{example}
Consider a Principal whose company's activity temporarily decreases and who therefore has to underwork by half the Agent. Then the loss in expected utility incurred by the Principal, no matter the wage he pays, is quantified by :
$$ R(a)^{\frac{\gamma_A + \gamma_P}{\gamma_A}}  =  \frac{\exp(-\gamma_P( \frac{a^*}{2} - \kappa(\frac{a^*}{2})))}{\exp(-\gamma_P( a^* - \kappa(a^*))}$$
\end{example}

Of course, we see that the impact of a choice of action depends  in turn on the cost of action function $\kappa.$ For example for a function $\kappa$ that is convex yet close to linear, choosing a non-optimal action will not have as much as an effect on the Principal as when $\kappa$ is convex and quadratic. Naturally, we may further analyze the effect of $\kappa$ and wish to quantify the effect of two different cost functions. We do so in the following and the analysis is similar to above. 

\begin{prop}[Effect of the action cost function $\kappa$]
Let $(W^*,a^*)$ be the optimum obtained in Theorem \ref{theo:main2} for the Risk-Sharing problem and a given cost function $\kappa$. Let $(W,a)$ be some contract in $C^{adm}$ for some cost function $\hat{\kappa}$ and binding the Participation Constraint. Then  : 

\begin{align}
\label{eq:ineqK}
\frac{\E\left[U_P\left(X^{a} - W\right)\right] }{ \E\left[U_P\left(X^{a^*} - W^*\right)\right]} &\geq \frac{\E\left[\left|U_P\left(X^a -  \hat{\kappa}(a) \right)\right|^{\frac{\gamma_A}{\gamma_A + \gamma_P}}\right]^{\frac{\gamma_A+\gamma_P}{\gamma_A}}}{\E\left[\left|U_P\left(X^{a^*} -  \kappa(a^*) \right)\right|^{\frac{\gamma_A}{\gamma_A + \gamma_P}}\right]^{\frac{\gamma_A+\gamma_P}{\gamma_A}} } \end{align}

In particular if $W$ is of the form $\frac{\gamma_P}{\gamma_A+\gamma_P} X^{a} + \beta$ then the above inequality holds in equality. Thus, if $(W,a)$ is the optimum for some particular $\hat{\kappa}$, one can quantify the effect of $\hat{\kappa}$ on the obtained optimum relative to that obtained for $\kappa$. Indeed, it holds that : 
 
 \begin{align}
\label{eq:ineqK}
\E\left[U_P\left(X^{a} - W\right)\right]  =\frac{\E\left[\left|U_P\left(X^a -  \hat{\kappa}(a) \right)\right|^{\frac{\gamma_A}{\gamma_A + \gamma_P}}\right]^{\frac{\gamma_A+\gamma_P}{\gamma_A}}}{\E\left[\left|U_P\left(X^{a^*} -  \kappa(a^*) \right)\right|^{\frac{\gamma_A}{\gamma_A + \gamma_P}}\right]^{\frac{\gamma_A+\gamma_P}{\gamma_A}} }  \times \E\left[U_P\left(X^{a^*} - W^*\right)\right]. 
 \end{align}
 
\end{prop}
\begin{proof}
Using similar reasoning to that in the proof of Proposition \ref{prop:effa}, we have that :
\begin{align}
\frac{\E\left[U_P\left(X^{a} - W\right)\right] }{ \E\left[U_P\left(X^{a^*} - W^*\right)\right]} &\geq  \frac{\E\left[\left|U_P\left(X^a -  \hat{\kappa}(a) \right)\right|^{\frac{\gamma_A}{\gamma_A + \gamma_P}}\right]^{\frac{\gamma_A+\gamma_P}{\gamma_A}}}{\E\left[\left|U_P\left(X^{a^*} -  \kappa(a^*) \right)\right|^{\frac{\gamma_A}{\gamma_A + \gamma_P}}\right]^{\frac{\gamma_A+\gamma_P}{\gamma_A}}} \times \left(  \frac{\E\left[ U_A\left(W-  \hat{\kappa}(a) \right)\right]}{ U_A(y_0)} \right)^{-\frac{\gamma_P}{\gamma_A}},
 \end{align}

and in particular if $(W,a)$ binds the Participation Constraint then 
$$ \left(  \frac{\E\left[ U_A\left(W-  \hat{\kappa}(a) \right)\right]}{ U_A(y_0)} \right)^{-\frac{\gamma_P}{\gamma_A}}=1,$$
and using the Reverse-Holder equality condition, we have equality for wages of the form $\frac{\gamma_P}{\gamma_A+\gamma_P} X^a + \beta$ saturating the Participation Constraint. \\
\end{proof}

This Proposition underlines the effect of a change in action cost function on the Risk-Sharing optimum, encompassing it in the action cost ratio. A corollary to this Proposition extends this result to the situation where a Principal has a choice between two Agents with respective cost functions $\kappa$ and $\hat{\kappa}$, and respective risk aversion coefficients $\gamma_A$ and $\hat{\gamma_A}$, and wishes to choose one and provide him with the related optimal Risk-Sharing contract.

\begin{cor}[A comparison of two Agents] Consider a Principal who has a choice between two agents characterized by $(\gamma_A, \kappa, y_0)$ and $(\hat{\gamma_A}, \hat{\kappa}, \hat{y}_0)$. Suppose that $\gamma_A y_0 = \hat{\gamma_A} \hat{y}_0$. Let $(W,a)$ and $(\hat{W}, \hat{a})$ be the associated optimal contracts (recalling that $a = \kappa^*(1)$ and $\hat{a} = \hat{\kappa}^*(1)$). Then :

$$ \E\left[U_P\left(X^{a} - W\right)\right] =   \frac{\E\left[\left|U_P\left(X^a -  \kappa(a) \right)\right|^{\frac{\gamma_A}{\gamma_A + \gamma_P}}\right]^{\frac{\gamma_A+\gamma_P}{\gamma_A}}}{\E\left[\left|U_P\left(X^{\hat{a}} -  \hat{\kappa}(\hat{a}) \right)\right|^{\frac{\hat{\gamma_A}}{\hat{\gamma_A} + \gamma_P}}\right]^{\frac{\hat{\gamma}_A+\gamma_P}{\hat{\gamma}_A}}} \times \E\left[U_P\left(X^{\hat{a}} - \hat{W}\right)\right]. $$
Therefore if  $$ \frac{\E\left[\left|U_P\left(X^a -  \kappa(a) \right)\right|^{\frac{\gamma_A}{\gamma_A + \gamma_P}}\right]^{\frac{\gamma_A+\gamma_P}{\gamma_A}}}{\E\left[\left|U_P\left(X^{\hat{a}} -  \hat{\kappa}(\hat{a}) \right)\right|^{\frac{\hat{\gamma_A}}{\hat{\gamma_A} + \gamma_P}}\right]^{\frac{\hat{\gamma}_A+\gamma_P}{\hat{\gamma}_A}}} < 1,$$
the Principal should choose the Agent characterized by $(\gamma_A, \kappa)$ and vice-versa.
\end{cor}

The choice of Agent therefore depends on a balance between the Agent's risk aversion and his action cost function. Note that this result may be generalized beyond the case where $\gamma_A y_0 = \hat{\gamma_A} \hat{y}_0$, one simply has to include the effect of $y_0$ and $\hat{y}_0$ in the comparison criteria. 

\begin{example}
Consider $B \sim \mathcal{N}(0,1),$ $\kappa(x) := \frac{x^2}{2}$ and $\hat{\kappa}(x) := x^2$. Then $a=1$ and $\hat{a} = \frac{1}{2}$. The Principal's expected utility obtained with the first Agent would be : 
\begin{align*}
\E[|U_P(X^a - \kappa(a))|^\frac{\gamma_A}{\gamma_A+\gamma_P}]^{\frac{\gamma_A+\gamma_P}{\gamma_A}} &= \exp\left(-\gamma_P(x_0 + a - \kappa(a)) + \frac{\gamma_A \gamma_P^2}{\gamma_A + \gamma_P}\right) \\
&= \exp\left(-\gamma_P\left(x_0 + \frac{1}{2} - \frac{\gamma_A \gamma_P}{\gamma_A + \gamma_P} \right)\right). \\
\end{align*}

Similarly, his expected utility from the second Agent would be: 

\begin{align*}
\E[|U_P(X^{\hat{a}} - \kappa(\hat{a}))|^\frac{\hat{\gamma}_A}{\hat{\gamma}_A+\gamma_P}]^{\frac{\hat{\gamma}_A+\gamma_P}{\hat{\gamma}_A}} &= \exp\left(-\gamma_P\left(x_0 + \frac{1}{4} - \frac{\hat{\gamma}_A \gamma_P^2}{\hat{\gamma}_A + \gamma_P}\right)\right). \\
\end{align*}

Therefore in order to maximize his utility, the Principal should employ the first Agent if $$ \frac{1}{4} - \frac{\gamma_A \gamma_P^2}{\gamma_A + \gamma_P} >  -\frac{\hat{\gamma}_A \gamma_P^2}{\hat{\gamma}_A + \gamma_P},$$
and he should employ the second Agent if not. Note that for more general cost functions and a Gaussian $B$, the Principal should employ the first Agent if 
$$ a - \kappa(a) - \frac{\gamma_A \gamma_P^2}{\gamma_A + \gamma_P} > \hat{a} - \hat{\kappa}(\hat{a})-\frac{\hat{\gamma}_A \gamma_P^2}{\hat{\gamma}_A + \gamma_P}. $$
\end{example}

\begin{remark}
In the same vein as the Propositions above, one may also exploit the consequences to Reverse-Hölder decomposition of the Principal's utility to compare the optimal Risk-Sharing contract to the optimal Moral Hazard contract, under CARA utilities. 
\end{remark}

\section{Extensions to the Reverse-Hölder framework}
\label{section:exten}

In the following Section, we illustrate the versatility of our framework and our results. Indeed so far we have worked with a single period model with a general production process and a general action cost function. We obtain existence, uniqueness and characterization of the Risk-Sharing solution, and some analysis on the model. In the following, we discuss on some extensions to this setting.  One first important setting is the continuous time setting.

\subsection{Existence, uniqueness and characterization of the Risk-Sharing optimum in a continuous time setting}

%

We specify the model of interest which is simply a continuous time version of the one studied previously. More precisely we consider one Principal and one Agent. The Principal provides a single cash flow (wage) $W$ at maturity (denoted $T$) to the Agent and requires in exchange an action $a=(a_t)_{t\in [0,T]}$ (that is completely monitored by the Principal) continuously in time according to the random fluctuations of the wealth of the firm. A contract will once again be a pair (wage,action)$=(W,a)$. \\

We start with the probabilistic structure that is required to define the random fluctuations of the wealth of the Principal. 
Let  $(\Omega,\mathcal{F},\P)$ be a probability space on which a stochastic process $B:=(B_{t})_{{t \in [0,T]}}$ is defined with its natural and completed filtration $\mathbb{F}:=(\mathcal{F}_t)_{t\in [0,T]}$. The only requirement on this stochastic process is that 
$$ \sup_{t\in [0,T]} \E\left[\exp\left(q  B_t\right)\right]<+\infty, \; \forall q \in \mathbb{R}^*.$$
This allows for the use of CARA utilities, and is verified for example for Brownian motion. We denote by $\E[\cdot]$ the expectation with respect to the probability measure $\P$.\\\\
\noindent
The Agent will be asked to perform an action $a$ continuously in time, according to the performances of the firm. Hence we introduce the set $\mathcal{P}$ of $\mathbb{F}$-predictable stochastic processes $a=(a_t)_{t\in [0,T]}$ and the set of actions is given as : 
$$\mathbb{H}_{2}:=\left\{ a=(a_{t})_{t\in[0,T]} \in \mathcal{P}, \; \textrm{ s.t. } \E\left[\exp\left(q\int_{0}^{T} |a_{t}|^{2} dt\right)\right]<+\infty, \; \forall q>0 \right\}.$$
As we will work with exponential preferences for the Agent and the Principal, we require so-called "exponential moments" on the actions and wages. As mentioned previously, this is a technical assumption.  
Given $a$ in $\mathbb{H}_{2}$, the wealth of the principal at any intermediate time $t$ between $0$ and the maturity $T$ is given by : 
\begin{equation}\label{eq:dynX}
X_t^{a} = x_0 + \int_{0}^{t} a_{s} ds + B_t, \quad t\in [0,T],\; \P-a.s.,
\end{equation}
where $x_0\in\R$ is a fixed real number. For any action $a$ in $\mathbb{H}_{2}$, we set $\mathbb{F}^{X^a}:=(\mathcal{F}_t^{X^a})_{t\in [0,T]}$ the natural filtration generated by $X^a$. In particular, we are interested in the set of $\mathcal{F}^{X^a}_T$-measurable random variables which provides the natural set for the wage $W$ paid by the Principal to the Agent. More precisely, we set\footnote{$\mathbb{R}^*:=\mathbb{R}\setminus\{0\}$} 
$$\mathcal{W}:=\left\{\mathcal{F}^{X^a}_{T}-\textrm{measurable random variables } W, \; \E[\exp(q \, W)]<+\infty, \; \forall q \in \mathbb{R}^* \right\}.$$
The fact that we ask for so-called finite exponential moments of any (positive, respectively negative) order for the action (respectively for the wage) is purely technical. As we will see, the optimal contract will satisfy these technical assumptions.\\\\
\noindent
The cost of effort for the Agent is once again modeled by a convex, continuous and non-decreasing function $\kappa$ defined on $\R_+$. As explained in the introduction for the single period problem, the Agent will accept a contract $(W,a)$ in $\mathcal{C}$ if and only if the following participation constraint (PC) is satisfied :
\begin{equation}
\label{eq:PC}
\E\left[U_{A}\left(W-\int_{0}^{T} \kappa(a_{t}) dt\right)\right] \geq U_{A}(y_0),
\end{equation}
where : $y_0$ is a given real number, $\kappa : \mathbb{R^+} \to \mathbb{R}$ models the cost of effort for the agent and is as discussed above, and $U_{A}(x):=-\exp(-\gamma_{A})(x)$ with $\gamma_{A}>0$ the risk aversion parameter for the Agent. From now on we assume that parameters $(y,\gamma_{A})$ are fixed. 
With these notations at hand, we can state the Principal's problem which writes down in term of a classical first-best problem as follows: 
\begin{equation}
\label{eq:OptiPb}
\sup_{(W,a) \in C^{adm}} \E\left[U_{P}\left(X_{T}^{a}-W\right)\right],
\end{equation}
where $U_{P}(x):=-\exp(-\gamma_{P}x)$ with $\gamma_{P}>0$ fixed and where 
$$ C^{adm}:=\left\{(W,a) \in \mathcal{C} \times \mathbb{H}_2, \; (\ref{eq:PC}) \textrm{ is in force}\right\}$$
is the set of admissible contracts satisfying the participation constraint (\ref{eq:PC}). We note that this continuous time problem may be dealt with using tools from the field of optimal stochastic control. Indeed, one may exploit  Equation (\ref{eq:PC}) to obtain a parametrization of all the wages satisfying the Participation Constraint for a given action process $(a_t)_{t \in [0,T]}$. This allows us to rewrite Problem (\ref{eq:OptiPb}) as a standard optimal control problem and one way of then proving existence involves using verification results as long as the required hypotheses are verified. \\ 

In the following we provide the continuous time counterpart to Theorem \ref{theo:main2}. It provides existence, uniqueness and characterization under quite general hypotheses (notably for a general process $(B_t)_{t \in [0,T]}$). We believe that this theorem brings to light the structure of the underlying problem and thus complements possible existence and uniqueness results exploiting optimal control.  

\begin{theorem}[Existence, uniqueness and characterization]
\label{theo:main2cont}
Consider the contract $(W^*, a^*)$ defined by setting :  $$a^*_t = \kappa^*(1)$$ for any $t$ in $[0,T]$ and 
 $$W^* = \frac{\gamma_P}{\gamma_P + \gamma_A} X^{a^*}_T + \beta^*,$$ 
 where the constant $\beta^*$ is worth :
$$ \beta^*:= y_0+ T\kappa(\kappa^*(1)) - \frac{\gamma_P}{\gamma_A+\gamma_P} \left(x_0+ T \kappa^*(1) \right) - \frac{1}{\gamma_A}\ln\left( \mathbb{E}\left[\exp\left(\frac{-\gamma_P \gamma_A}{\gamma_P+\gamma_A} \int_0^T dB_t\right)\right] \right).$$

Then 
$(W^*, a^*)$ both satisfies and saturates the participation constraint. Furthermore for any $(W,a)$ in $C^{adm}$
$$\E\left[U_P\left(X_T^{a} - W\right)\right] \leq \E\left[U_P\left(X_T^{a^*} - W^*\right)\right],$$ 
with equality only for $(W^*, a^*)$. $(W^*, a^*)$ is therefore the unique optimal contract in the continuous-time Risk-Sharing problem. Finally, the Principal's optimal utility is :
 \begin{align*}
& \E\left[U_P\left(X_T^{a^*} - W^*\right)\right] =  U_P(x_0-y_0 - \gamma_P T \widetilde{\kappa}(1) ) \E\left[  \exp\left(\frac{-\gamma_P \gamma_A}{\gamma_P+\gamma_A} \int_0^T dB_t \right)\right]^{\frac{\gamma_A + \gamma_P}{\gamma_A}}.
\end{align*}

\end{theorem}
\begin{proof}
The proof closely mirrors that of Theorem \ref{theo:main2} with the calculations done in continuous time.
\end{proof}
It is important to note that, as a consequence to this Theorem, the analysis conducted in Section \ref{section:an} nicely generalizes to this setting and allows once again for the comparison of contracts.  In this continuous time setting one may even easily analyze the effect of choosing a non-constant action $a$. 
\begin{remark} 
In the widely studied case where $B$ is in fact a standard Brownian motion, Theorem \ref{theo:main2cont} applies and gives existence and uniqueness. The optimal $\beta^*$ has the expression :
$$ \beta^*:= \frac{T \gamma_A |\gamma_P|^2}{2|\gamma_A+\gamma_P|^2} + y_0+ T\kappa(\kappa^*(1)) - \frac{\gamma_P}{\gamma_A+\gamma_P} \left(x_0+ T \kappa^*(1) \right),$$
and the Principal's expected utility is :
$$ \E\left[U_P\left(X_T^{a^*} - W^*\right)\right] =  U_P(x_0-y_0) \exp\left(\gamma_P T \left(-\widetilde{\kappa}(1) + \frac{\gamma_P \gamma_A}{2(\gamma_A+\gamma_P)} \right) \right).$$
We therefore complement the characterization work of Muller in \cite{Muller98} by providing an existence and uniqueness proof, as well as providing an alternative characterization method.
\end{remark}

\begin{remark}
The Principal-Agent problem has also been analyzed in the setting where a Principal employs several different Agents. Note that for CARA Agents this framework also works. 
\end{remark}

We now turn to the case of a risk neutral Principal which can be analyzed as a limit of the case of a CARA Principal, and we do so in the following.

\subsection{The case of a risk-neutral Principal}
\label{sectionRN}

The analysis provided throughout this paper concerns a Principal and an Agent who are both risk averse with the risk aversion modeled through the CARA utility functions. In fact, the key to this paper is the exponential properties of the CARA utilities. However, an important case in the literature is that of a risk neutral Principal who wishes to employ a risk averse Agent. More precisely, if we for example set ourselves in the discrete-time setting, the Risk-Sharing problem (\ref{eq:ProbRS}) becomes : 

\begin{equation}
\label{eq:OptiPbBis}
\sup_{(W,a) \in \mathcal{C}^{adm}} \E\left[X^{a}-W\right],
\end{equation}  
where we use the same notations as previously (in particular the Participation Constraint is in force for the Agent with utility function $U_A(x)=-\exp(-\gamma_A x)$). Since our Reverse Hölder approach relies on the structure of functions $U_P$ and $U_A$, we cannot carry it directly in the risk neutral case. However, as it is well-known, the risk neutral framework can be seen as a limit case with formally $\gamma_P=0$ by rescaling the mapping $U_P$ to become $\tilde U_P(x):=-\frac{\exp(-\gamma_P x) -1}{\gamma_P}$ and by letting $\gamma_P$ go to $0$. Hence, we can use our approach with $\tilde U_P$ and $U_A$ to obtain existence of an optimum and its characterization in the risk-neutral case. \\

\noindent
Consider a contract $(W,a)$ that satisfies the (PC). Then by Lemma \ref{Lemma:UI} (in Section \ref{section:technic}), 
\begin{align*}
\E\left[X^{a}-W\right] &= \E \left[ \lim_{\gamma_P \rightarrow 0 }  \tilde{U}_P(X^a-W) \right]\\
&= \lim_{\gamma_P \rightarrow 0}  \E \left[   \tilde{U}_P(X^a-W) \right]\\
&= \lim_{\gamma_P \rightarrow 0} \gamma_P^{-1} \left( \E[U_P (X^a - W)] +1\right)\\
&\leq \lim_{\gamma_P \rightarrow 0} \gamma_P^{-1} \left( \E[U_P (X^{a^*} - W^*)] +1\right),
\end{align*}
according to (ii) of Theorem \ref{theo:main2}. Using the explicit computation of the upper bound's value, we have that
\begin{align*}
\E\left[X^{a}-W\right] &= \E \left[ \lim_{\gamma_P \rightarrow 0 }  \tilde{U}_P(X^a-W) \right]\\
&\leq \lim_{\gamma_P \rightarrow 0} \gamma_P^{-1} \left(  U_P(x_0-y_0 + \widetilde{\kappa}(1))  \E\left[  \exp\left(\frac{-\gamma_P \gamma_A}{\gamma_P+\gamma_A} B \right)\right]^{\frac{\gamma_A + \gamma_P}{\gamma_A}}+1\right)\\
&= x_0 - y_0 + \widetilde{\kappa}(1).
\end{align*}
So we have given the upper bound $x_0 - y_0 + \widetilde{\kappa}(1)$ to the value problem of the Risk Neutral Principal. An explicit computation gives that this upper bound can be attained by choosing the contract $(W^*_{RN}, a^*_{RN})$ with
 $$a^*_{RN} = \kappa^*(1) \text{ and } W^*_{RN}=y_0 + \kappa(\kappa^*(1)),$$
which is formally the optimal contract found in in Theorem \ref{theo:main2} with $\gamma_P=0$. The optimal parameters have economic meaning : the Principal is neutral to risk and is thus willing to give a fixed wage to his Agent regardless of the performance of the output process. We note that in this case, the Risk-Sharing structure of the problem disappears and the Principal carries all of the risk. 
\\\\
\noindent 
We thus have an existence proof for the risk-neutral case, along with a possible characterization of the optimum. Of course we may perform the computation of the risk-neutral optimum more directly without exploiting the risk-averse optimum provided through the method used in this paper. However, the fact that the risk-neutral case may be seen as a limit of the risk-averse case allows for example for the extension of the results of Section \ref{section:an} to neutral settings. We give such an extension in the following. 

\begin{prop}
\label{prop:RNana}
Let $(W^*, a^*)$ be the risk-averse optimum for $\hat{U}_P(x) = -\exp(-\hat{\gamma}_P x)$ given in Theorem \ref{theo:main2} and let $(W,a)$ be in $C^{adm}$. Then : 
$$\frac{\E\left[X^{a} - W\right] }{ \E\left[\hat{U}_P\left(X^{a^*} - W^*\right)\right]} \geq \lim_{\gamma_P \rightarrow 0} \gamma_P^{-1}  \left( \frac{ \E\left[\left|U_P\left(X^a -  \kappa(a) \right)\right|^{\frac{\gamma_A}{\gamma_A + \gamma_P}}\right]^{\frac{\gamma_A+\gamma_P}{\gamma_A}} }{\E\left[\left|\hat{U}_P\left(X^{a^*} -  \kappa(a^*) \right)\right|^{\frac{\gamma_A}{\gamma_A + \hat{\gamma}_P}}\right]^{\frac{\gamma_A+ \hat{\gamma}_P}{\gamma_A}} } \times \frac{U_P(y_0)}{\hat{U}_P(y_0)} \right).$$
\end{prop}

\begin{proof}
See Section \ref{section:proofRNana}.
\end{proof}

This Proposition is an asymptotic form of Proposition \ref{prop:effa} for comparison between a risk-neutral and risk-averse Principal : we obtain a decomposition dependent on  an action ratio and a participation constraint ratio. Note however that the risk-neutral utility is not signed and this inequality is thus more difficult to exploit. Finally, note that this analysis also holds in continuous time.

\section{Proofs}
\label{section:proofs}
 
 In this Section we collect the proofs of the technical results we made use of to proceed with our analysis. 
 
 \subsection{Proof of Theorem \ref{th:main1}}
 \label{section:proofdecomp}

\begin{proof}
We fix $a$ in $\R_+$ and prove each item of the Theorem. 
\begin{itemize}
\item[(i)] For the first result, we express the (expected) utility of the Principal in terms of the one of the Agent. We have : 
\begin{align}
\label{eq:temp1}
&\E\left[U_P\left(X^{a} - W\right)\right] \nonumber\\
&=\E\left[U_P\left(X^a -   \kappa(a)\right) \times \exp\left( \gamma_P \left(W- \kappa(a) \right) \right)\right]\nonumber\\
&=\E\left[U_P\left(X^a -   \kappa(a) \right) \times \left|U_A\left(W-\kappa(a) \right)\right|^{-\frac{\gamma_P}{\gamma_A}}\right].
\end{align}
We wish to extract the Agent's utility from this expression and obtain at least an inequality. To do so, we need some kind of H\"older inequality. However the classical H\"older inequality cannot be applied for two reasons : first the exponent $-\frac{\gamma_P}{\gamma_A}$ of the utility of the Agent is negative; and then the negativity of the mapping $U_P$ calls for the use of a H\"older inequality in the reverse direction. These two features are taken into account in the so-called \textit{ Reverse H\"older inequality } which can be seen as a counterpart to the classical H\"older inequality and given in Proposition  \ref{prop:RH}. In particular, we wish to use Item (i). More precisely, let : 
\begin{equation}
\label{eq:FG}
F := U_P\left(X^a - \kappa(a) \right), \quad G := \left|U_A\left(W-  \kappa(a) \right)\right|^{-\frac{\gamma_P}{\gamma_A}}.
\end{equation}
Note naturally that these two random variables depend on the contract $(W, a)$ under interest.\\\\
\noindent 
We wish to apply Reverse-H\"older to $F$ and $G$ with some exponent $p$ that we calibrate so that  $|G|^{-\frac{1}{p-1}} = \left|U_A\left(W-  \kappa(a) \right)\right|$; which immediately gives  $p = 1 + \frac{\gamma_P}{\gamma_A} = \frac{\gamma_A + \gamma_P}{\gamma_A} > 1$.
We thus immediately obtain:
$$ \E\left[|F|^{\frac1p}\right] = \E\left[\left|U_P\left(X^a -  \kappa(a) \right)\right|^{\frac{\gamma_A}{\gamma_P + \gamma_A}}\right],$$
$$\E\left[|G|^{\frac{-1}{p-1}} \right] = -\E\left[ U_A\left(W -  \kappa(a) dt \right)\right].$$ 
Applying (i) of Proposition \ref{prop:RH} to $F$ and $G$ with this particular choice of $p$ in Relation (\ref{eq:temp1}) gives our result : 
\begin{align}
\label{eq:temp2}
&\E\left[U_P\left(X_T^{a} - W\right)\right] \nonumber\\
&=-\E\left[\left|F G\right|\right] \nonumber\\
&\leq \E\left[\left|U_P\left(X^a -\kappa(a)\right)\right|^{\frac{\gamma_A}{\gamma_A + \gamma_P}}\right]^{\frac{\gamma_A+\gamma_P}{\gamma_A}} \times \E\left[ U_A\left(W- \kappa(a) \right)\right]^{-\frac{\gamma_P}{\gamma_A}}.
\end{align} 


\item[(ii)]  We first prove that (ii\textquotesingle) is equivalent to (ii\textquotesingle\textquotesingle). This involves finding an equality condition for (\ref{eq:temp2}). Through (ii) of Proposition \ref{prop:RH}, Inequality (\ref{eq:temp2}) is an equality if and only the contract $(a,W)$ is such that there exists a positive constant $\alpha$ such that the random variables $F$ and $G$ defined in (\ref{eq:FG}) enjoys :
$$ |F| = \alpha |G|^{-\frac{p}{p-1}}.$$
By definition of $F$, $G$ and $p=\frac{\gamma_A+\gamma_P}{\gamma_A}$ this condition reads as : 
\begin{equation*}
\frac{\left|U_P\left(X^a - \kappa(a) \right) \right|}{\left| U_A\left(W- \kappa(a) dt \right) \right|^\frac{-( \gamma_P + \gamma_A)}{\gamma_A}} = \alpha 
\end{equation*}
Thus using the exponential form of our utilities we obtain the condition : 
\begin{align*}
\frac{\left|U_P\left(X_T^a - W \right) \right|}{\left| U_A\left(W-\kappa(a) \right) \right|} = \alpha, \\
\end{align*}
which is equivalent to 
\begin{align*}
\frac{U_P'\left(X_T^a - W\right)}{U_A'\left(W -  \int_0^T \kappa(a_t) dt \right)} = \frac{ \gamma_P}{\gamma_A} \alpha. 
\end{align*}
Setting $\alpha$ to $ \frac{ \gamma_P}{\gamma_A} \alpha$ we obtain our result. \\

We now prove that (ii\textquotesingle\textquotesingle) is equivalent to (ii\textquotesingle\textquotesingle\textquotesingle). When (ii\textquotesingle\textquotesingle) holds, $(W, a)$ satisfies (\ref{eq:PC}) and we have the following series of implications where $\alpha$ is a positive constant :  
\begin{align*}
& \quad \quad  \frac{U_P'\left(X^a -W\right)}{U_A'\left(W-  \kappa(a)\right)}=\alpha \\
&\Rightarrow (\gamma_P + \gamma_A)W - \gamma_P X^a - \gamma_A \kappa(a)   = \ln\left( \alpha\frac{\gamma_A}{\gamma_P}\right)\\
&\Rightarrow W = \frac{\gamma_P}{\gamma_P + \gamma_A} X^a + \frac{\gamma_A}{\gamma_P + \gamma_A} \kappa(a) +  \ln\left( \alpha\frac{\gamma_A}{\gamma_P}\right),\\
&\Rightarrow W = \frac{\gamma_P}{\gamma_P + \gamma_A} X^a + \beta,
\end{align*}
where $\beta =\frac{\gamma_A}{\gamma_P + \gamma_A}  \kappa(a)  +  \ln\left( \alpha\frac{\gamma_A}{\gamma_P}\right)$.  
%

Conversely, let us suppose that (ii\textquotesingle\textquotesingle\textquotesingle) holds. Then $(W, a)$ where $W = \frac{\gamma_P}{\gamma_P + \gamma_A} X^a + \beta$ satisfies (\ref{eq:PC}) and we have that : 
$$ \frac{U_P'\left(X^a -W\right)}{U_A'\left(W-  \kappa(a) \right)} = \exp\left( (\gamma_P + \gamma_A) \beta + \gamma_A  \kappa(a)  \right) \in \R^*_+.$$

\end{itemize}
 \end{proof}

\subsection{Proof of Proposition \ref{prop:upperbound}}
\label{section:proofPropupperBound}

Let $a \in \R_+$. We have
\begin{align*}
&\E\left[  \left|U_P\left(X^{a} - \kappa(a)  \right)\right|^{\frac{\gamma_A}{\gamma_A + \gamma_P}}\right]\\
&=\E\left[  \exp\left(-\gamma_P\left(X^{a} - \kappa(a)  \right)\right)^{\frac{\gamma_A}{\gamma_A + \gamma_P}}\right]\\
&=\E\left[  \exp\left(-\gamma_P\left(x_0 + \left(a-\kappa(a) \right) + B\right)\right)^{\frac{\gamma_A}{\gamma_A + \gamma_P}}\right]\\
&= \exp\left(-\frac{\gamma_P \gamma_A}{\gamma_A+\gamma_P} x_0 + \Phi(a) \right)  \E\left[  \exp\left(\frac{-\gamma_P \gamma_A}{\gamma_P+\gamma_A} B \right)\right]\\
\end{align*}
where 
$$ c\mapsto \Phi(c):=  -\frac{\gamma_P \gamma_A}{\gamma_A+\gamma_P} \left(c-\kappa(c)\right) .$$
Note that the mapping $\Phi$ is convex on $\mathbb{R}_+$, and letting $a^*:=\kappa^*(1)$, 
$$ \Phi(c) \geq \Phi(a^*) = -\frac{\gamma_P \gamma_A \widetilde{\kappa}(1)}{\gamma_A+\gamma_P}, \quad \forall c\geq0.$$
So, 
\begin{align*}
&\E\left[  \left|U_P\left(X^{a} -  \kappa(a) \right)\right|^{\frac{\gamma_A}{\gamma_P + \gamma_A}}\right]\\
&\geq \exp\left(-\frac{\gamma_P \gamma_A}{\gamma_A+\gamma_P} (x_0 + \tilde{\kappa}(1))\right)  \E\left[  \exp\left(\frac{-\gamma_P \gamma_A}{\gamma_P+\gamma_A} B \right)\right],
\end{align*}
and we thus deduce our result.

\subsection{Proof of the optimal contract : Theorem \ref{theo:main2}}
\label{section:ProofOpt}

We consider the contract $(W^*, a^*)$ defined by setting :  $$a^* := \kappa^*(1)$$ and 
 $$W^* = \frac{\gamma_P}{\gamma_P + \gamma_A} X^{a^*} + \beta^*,$$ 
$$ \beta^*:=  y_0+ \kappa(\kappa^*(1)) - \frac{\gamma_P}{\gamma_A+\gamma_P} \left(x_0+  \kappa^*(1) \right) - \frac{1}{\gamma_A}\ln\left( \mathbb{E}\left[\exp\left(\frac{-\gamma_P \gamma_A}{\gamma_P+\gamma_A} B\right)\right] \right).$$
We first study the participation constraint to verify the admissibility of such a contract. We have that : 
\begin{align*}
& \E\left[U_A\left( W^* -  \kappa(a^*) \right)\right] = \E\left[U_A\left(  \frac{\gamma_P}{\gamma_P + \gamma_A} X^{a^*}+ \beta^* +  \kappa(\kappa^*(1)) \right)\right] \\
&= \E\left[U_A\left(  \frac{\gamma_P}{\gamma_P + \gamma_A} (x_0 + \kappa^*(1) + B ) + \beta^* - \kappa(\kappa^*(1)) \right)\right]\\
&= U_A\left(  y_0\right)\E\left[\exp \left(  \frac{-\gamma_A \gamma_P}{\gamma_P + \gamma_A} B - \ln\left( \mathbb{E}\left[\exp\left(\frac{-\gamma_P \gamma_A}{\gamma_P+\gamma_A} B\right)\right]\right)
 \right)\right] = U_A(y_0). 
\end{align*}
Thus $W^*$ belongs to $\mathcal{W}(a^*)$. According to Item (ii) of Theorem \ref{th:main1}, the contract satisfies the Borch Rule. Furthermore, it is of the form $(W^*, \kappa^*(1))$ where $W^*$ saturates the Participation Constraint. It follows that the equality conditions to reach the upper bound of the Principal's Expected Utility are verifies and we have that for any $a$ in $\R_+$ and any $W$ in $\mathcal{W}^*(a)$, 

 \begin{align*}
 &\E\left[U_P\left(X^{a} - W\right)\right] \\
& \leq \E\left[U_P\left(X^{a^*} - W^*\right)\right] =  U_P(x_0-y_0) \exp\left(\gamma_P T \left(-\widetilde{\kappa}(1) + \frac{\gamma_P \gamma_A}{2(\gamma_A+\gamma_P)} \right) \right).
\end{align*}

We deduce that $(W^*, a^*)$ is the optimal contract for the first-best Principal-Agent problem.

\subsection{A technical lemma}
\label{section:technic}
\begin{lemma}
\label{Lemma:UI}
Let $(W,a)$ be an admissible contract in $\mathcal{C}$. The sequence of random variables \\
$\left(\tilde U_P(X^a-W)\right)_{0<\gamma_P<1}$ is uniformly integrable. And so : 
$$ \E\left[X^a-W \right]=\E\left[\lim_{\gamma_P \to 0} \tilde U_P(X^a-W) \right] = \lim_{\gamma_P \to 0} \E\left[ \tilde U_P(X^a-W) \right].$$
\end{lemma}

\begin{proof}
The second part of the statement is a consequence of the uniform integrability (UI) and of the fact that the identity mapping is the limit (as $\gamma_P$ goes to $0$) of $\tilde U_P$. So we focus on the UI property and apply de la Vall\'ee-Poussin criterion. We have :
\begin{align*}
&\sup_{0<\gamma_P<1} \E\left[|\tilde U_P(X^a-W)|^2 \right]\\
&= \gamma_P^{-2} \sup_{0<\gamma_P<1} \E\left[|\exp(-\gamma_P(X^a-W))-1|^2\right]\\
&= \sup_{0<\gamma_P<1} \E\left[|\bar X|^2 |\exp(-\gamma_P \bar X)|^2\right],
\end{align*} 
where $\bar X$ is a random point between $0$ and $X^a-W$ (using mean value theorem). By Cauchy-Schwarz's inequality we have, 
\begin{align*}
&\sup_{0<\gamma_P<1} \E\left[|\tilde U_P(X^a-W)|^2 \right]\\
&\leq \E\left[|\bar X|^4\right]^{1/2} \sup_{0<\gamma_P<1} \E\left[\exp(-4 \gamma_P \bar X)\right]^{1/2}.
\end{align*} 
As $|\bar X| \leq |X^a-W|$, $\P$-a.s., we have that $\E\left[|\bar X|^4\right]<+\infty$. Regarding the second term, 
\begin{align*}
&\sup_{0<\gamma_P<1} \E\left[\exp(-4 \gamma_P \bar X)\right]\\
&\leq \sup_{0<\gamma_P<1} \left( \P\left[\bar X \geq 0\right] + \E\left[\exp(-4 \gamma_P \bar X) \textbf{1}_{\bar X <0}\right] \right)\\
&\leq 1 + \E\left[\exp(-4 R \bar X) \textbf{1}_{\bar X <0}\right] \\
&\leq 1 + \E\left[\exp(4 R |X^a-W|) \right] <+\infty.
\end{align*}
\end{proof}

\subsection{Proof of Proposition \ref{prop:RNana}}
\label{section:proofRNana}

Let $(W^*, a^*)$ be the risk-averse optimum for $\hat{U}_P(x) = -\exp(-\hat{\gamma}_P x)$ given in Theorem \ref{theo:main2} and let $(W,a)$ be in $C^{adm}$. Then as the expected utility function is negative, and using the Reverse-Holder inequality we obtain : 
\begin{align*}
\frac{\E\left[X^{a} - W\right] }{ \E\left[\hat{U}_P\left(X^{a^*} - W^*\right)\right]}&= \frac{\lim_{\gamma_P \rightarrow 0} \gamma_P^{-1} \left( \E[U_P (X^a - W)] +1\right)}{ \E\left[\hat{U}_P\left(X^{a^*} - W^*\right)\right]} \geq  \frac{\lim_{\gamma_P \rightarrow 0} \gamma_P^{-1} \left( \E[U_P (X^a - W)] \right)}{ \E\left[\hat{U}_P\left(X^{a^*} - W^*\right)\right]} \\\
& \geq \frac{\lim_{\gamma_P \rightarrow 0} \gamma_P^{-1} \left( \E\left[\left|U_P\left(X^a -  \kappa(a) \right)\right|^{\frac{\gamma_A}{\gamma_A + \gamma_P}}\right]^{\frac{\gamma_A+\gamma_P}{\gamma_A}} \times \E\left[ U_A\left(W-  \kappa(a) \right)\right]^{-\frac{\gamma_P}{\gamma_A}}\right) }{\E\left[\left|\hat{U}_P\left(X^{a^*} -  \kappa(a^*) \right)\right|^{\frac{\gamma_A}{\gamma_A + \hat{\gamma}_P}}\right]^{\frac{\gamma_A+\hat{\gamma}_P}{\gamma_A}}  \times \E\left[ U_A\left(W^*-  \kappa(a^*) \right)\right]^{-\frac{\hat{\gamma}_P}{\gamma_A}}}\\
& = \frac{\lim_{\gamma_P \rightarrow 0} \gamma_P^{-1} \left( \E\left[\left|U_P\left(X^a -  \kappa(a) \right)\right|^{\frac{\gamma_A}{\gamma_A + \gamma_P}}\right]^{\frac{\gamma_A+\gamma_P}{\gamma_A}} \times \E\left[ U_A\left(W-  \kappa(a) \right)\right]^{-\frac{\gamma_P}{\gamma_A}} \right) }{\E\left[\left|\hat{U}_P\left(X^{a^*} -  \kappa(a^*) \right)\right|^{\frac{\gamma_A}{\gamma_A + \hat{\gamma}_P}}\right]^{\frac{\gamma_A+\hat{\gamma}_P}{\gamma_A}}  \times \hat{U}_P(y_0)}\\
& \geq \frac{\lim_{\gamma_P \rightarrow 0} \gamma_P^{-1} \left( \E\left[\left|U_P\left(X^a -  \kappa(a) \right)\right|^{\frac{\gamma_A}{\gamma_A + \gamma_P}}\right]^{\frac{\gamma_A+\gamma_P}{\gamma_A}} \times U_P(y_0)  \right)  }{\E\left[\left|\hat{U}_P\left(X^{a^*} -  \kappa(a^*) \right)\right|^{\frac{\gamma_A}{\gamma_A + \hat{\gamma}_P}}\right]^{\frac{\gamma_A+\hat{\gamma}_P}{\gamma_A}}  \times \hat{U}_P(y_0)}\\
&= \lim_{\gamma_P \rightarrow 0} \gamma_P^{-1}  \left( \frac{ \E\left[\left|U_P\left(X^a -  \kappa(a) \right)\right|^{\frac{\gamma_A}{\gamma_A + \gamma_P}}\right]^{\frac{\gamma_A+\gamma_P}{\gamma_A}} }{\E\left[\left|\hat{U}_P\left(X^{a^*} -  \kappa(a^*) \right)\right|^{\frac{\gamma_A}{\gamma_A + \hat{\gamma}_P}}\right]^{\frac{\gamma_A+ \hat{\gamma}_P}{\gamma_A}} } \times \frac{U_P(y_0)}{\hat{U}_P(y_0)} \right).
\end{align*}

\newpage

\section{Conclusion}
\label{section:conc}

This paper uses the Reverse H\"older inequality to derive a new approach to the Risk-Sharing Principal-Agent problem. Through a specific decomposition of the Principal's expected utility (that relies of the multiplicative property of exponential utility functions) we are able to extract the participation constraint in its expectation form. We are then able to to prove existence and uniqueness of the optimal Risk-Sharing plan, whilst also characterizing the plan whilst and making the Borch rule appear. We note that this analysis allows for general hypothesis on the underlying model and works very similarly in both discrete and continuous time settings. It also extends to the risk-neutral case. As a by-product of this work, we are able to analyze the effect of enforcing a sub-optimal action and also provide some insight into the parameters affecting a Principal's choice between two Agents. Another natural extension to this analysis may be that of choosing a sub-optimal wage. Such a choice may make sense for many reasons such as wanting to enforce limited liability, and is the topic of ongoing research by the authors.

\section{Acknowledgements}
The authors wish to thank St\'ephane Villeneuve for insightful discussions and comments; and the ANR Pacman for financial support. 
 

\newpage
\begin{thebibliography}{1}


\bibitem{Borch62}
K.~Borch.
\newblock Equilibrium in a reinsurance market.
\newblock {\em Econometrica}, 30(3):424--444, 1962.

\bibitem{Borch74}
K.~Borch.
\newblock The mathematical theory of insurance an annotated selection of papers
  on insurance published 1960--1972.
\newblock {\em Lexington Books}, 1974.

\bibitem{Carlier01}
G.~Carlier.
\newblock A general existence result for the principal-agent problem with adverse selection.
\newblock Journal of Mathematical Economics, 2001.

\bibitem{Carlier19}
G.~Carlier, K.~Zhang
\newblock Existence of solutions to principal-agent problems under general preferences and non-compact allocation space.
\newblock 2019.


\bibitem{CZ12}
J. ~Cvitani\'c and J.  ~Zhang.
\newblock{\em Contract theory in continuous-time models}.
\newblock Springer Finance, 2012. 

\bibitem{CWZ05}
J. ~Cvitani\'c, X. ~Wan, and J.  ~Zhang.
\newblock{\em First-best contracts for continuous-time principal-agent problems}.
\newblock 2005. 

\bibitem{Embrechts17}
P. ~Embrechts, L. ~Haiyan and R. ~Wang.
\newblock{\em Quantile-Based Risk Sharing}.
\newblock Operations Research, 2018.  

\bibitem{HM87}
B. ~Holmstr\"om and P. ~Milgrom. 
\newblock {\em Aggregation and linearity in the provision of intertemporal incentives}.
\newblock Econometrica, 1987. 

\bibitem{LaffontMartimort}
J.~Laffont and D.~Martimort.
\newblock {\em Theory of incentives : the Principal-Agent model}.
\newblock Princeton University Press, 2009.

\bibitem{JKS}
Jewitt, I., Kadan, O.,  Swinkels, J. M., 
\newblock Moral hazard with bounded payments,
\newblock Journal of Economic Theory, 143(1), 59-82, 
\newblock 2008. 

\bibitem{KRS}
 Kadan, O.,  Reny, P., Swinkels, J. M., 
\newblock Existence of Optimal Mechanisms in Principal‐Agent Problemss,
\newblock Econometrica, 
\newblock 2017. 

\bibitem{Ligon11}
E. ~Ligon.
\newblock{\em Notes on Risk Sharing}.
\newblock University of Berkeley, 2011.

\bibitem{Marko52}
H. ~Markowitz.
\newblock{\em Portfolio Selection}.
\newblock The Journal of Finance, 1952. 

\bibitem{Moffet79}
D. ~Moffet.
\newblock{\em The Risk-Sharing Problem}.
\newblock The Geneva Papers on Risk and Insurance - Issues and Practice, 1979. 

\bibitem{Muller98}
H. ~Muller. 
\newblock{\em The First-Best Sharing Rule in the Continuous Time Principal-Agent Problem with Exponential Utility}.
\newblock Journal of Economic Theory, 1998. 

\bibitem{Page87}
F. ~Page. 
\newblock{\em The existence of optimal contracts in the principal-agent model}.
\newblock Journal of Mathematical Economics, 1987. 


\bibitem{Sannikov08}
Y. ~Sannikov. 
\newblock{\em A continuous-time version of the principal-agent problem}.
\newblock The Review of Economic Studies, 2008.

\bibitem{SSung93}
H. ~Sch\"attler and J. ~Sung.
\newblock {\em The first-order approach to the continuous-time principal-agent problem with exponential utility}.
\newblock Journal of Economic Theory, 1993.

\bibitem{SSung97}
H. ~Sch\"attler and J. ~Sung.
\newblock {\em On optimal sharing rules in discrete and continuous-time principal-agent problems with exponential utility}.
\newblock Journal of Economic Dynamics and Control, 1997.

\bibitem{Sung95}
J. ~Sung.
\newblock {\em Linearity with project selection and controllable diffusion rate in continuous-time principal-agent problems}.
\newblock The RAND Journal of Economics, 1995.






\end{thebibliography}

\end{document}